\documentclass[letterpaper, 10 pt, conference]{ieeeconf}  

\IEEEoverridecommandlockouts                              
\usepackage{dsfont}
\usepackage{graphicx}
\usepackage{epstopdf}
\usepackage{epsfig}
\usepackage{tikz}
\usepackage{rotating}
\usepackage{mathtools}
\usepackage{subfig}
\usepackage{enumerate,pgfplots}
\usepackage{hyperref}
\usepackage{amsmath,amssymb,amsfonts}
\usepackage{cleveref}
\usepackage{bbm}
\newtheorem{theorem}{Theorem}

\newtheorem{proposition}{Proposition}
\newtheorem{lemma}{Lemma}
\newtheorem{corollary}{Corollary}
\newtheorem{example}{Example}
\newtheorem{remark}{Remark}

\newcommand{\ba}{\begin{array}}
\newcommand{\ea}{\end{array}}

\newcommand{\be}{\begin{equation}}
\newcommand{\ee}{\end{equation}}

\newcommand{\ds}{\displaystyle}

\newcommand{\eps}{\varepsilon}

\newcommand{\mc}{\mathcal}

\newcommand{\ov}{\overline}
\newcommand{\ul}{\underline}

\def\1{\boldsymbol{1}}
\def\0{\boldsymbol{0}}
\def\qed{\hfill$\blacksquare$}

\newcommand{\R}{\mathbb{R}}

\newcommand{\C}{\mathbb{C}}

\newcommand{\inv}{{-1}}
\newcommand{\ga}{[\gamma]}
\newcommand{\de}{[\delta]}

\newcommand{\V}{\mathcal{V}}

\newcommand{\X}{\mathcal{X}}

\renewcommand{\Re}{\mathrm{Re}}

\def\R{\mathbb{R}}
\def\C{\mathbb{C}}

\tikzstyle{v_c}=[circle, draw,inner sep=2pt, minimum width=12pt, color=blue]
\tikzstyle{v_a}=[circle, draw,inner sep=2pt, minimum width=12pt, color=red]
\tikzstyle{edge} = [draw,thick,-,font=\small ]
\tikzstyle{label} = [draw,fill=black,font=\normalsize]

\begin{document}
\title{\LARGE \bf On the dynamic behavior of the network SIRS epidemic model}
\author{Giulia Gatti, Giacomo Como 
	\thanks{The authors are with the  Department of Mathematical Sciences ``G.L.~Lagrange,'' Politecnico di Torino, 10129 Torino, Italy. Email: giulia.gatti@polito.it, giacomo.como@polito.it }
	}
	
	\maketitle
	
	\begin{abstract}
We study the Susceptible-Infected-Recovered-Susceptible (SIRS) epidemic model on deterministic networks. For connected but otherwise general interaction patterns and heterogeneous recovery and loss-of-immunity rates, we identify a fundamental parameter $R_0$ (the basic reproduction number), which fully characterizes the qualitative dynamic behavior of the system. This parameter is the dominant eigenvalue of a rescaled version of the interaction matrix, whose rows are normalized by the corresponding recovery rates. We establish a sharp threshold phenomenon occurring as $R_0$ crosses the value $1$. Specifically, we show that, if $R_0\le1$, then the disease-free equilibrium is globally asymptotically stable, whereas, if $R_0>1$, then the disease-free equilibrium is unstable and there exists a unique endemic equilibrium, which is asymptotically stable. As a byproduct of our analysis, we also identify key monotonicity properties of the dependence of the endemic equilibrium on the model parameters (the interaction matrix as well as the recovery rates and the loss-of-immunity rates) and obtain a distributed iterative algorithm for its computation, with provable convergence guarantees. 
Our results extend existing ones available in the literature for network SIRS epidemic models with rank-one interaction matrices and homogeneous  recovery rates (including the single homogeneous population SIRS epidemic model).  
\end{abstract}


\section{Introduction}\label{sec:introduction}

The need to analyze and control the spread of infectious diseases has long been a driving force for the development of mathematical models for epidemics.  A fundamental class of such models consists of compartmental models, which describe the evolution in aggregate form by dividing the population into distinct compartments according to disease status, the most common of which are \textit{susceptible}, \textit{infected},  and \textit{recovered}.
A prototypical example is the Susceptible-Infected-Recovered  (SIR) model  \cite{Kermack.McKendrick:1927,KermackMcKendrick1932
}, which is extensively studied and applied in epidemiology \cite{anderson-may1991,Diekmann2000MathematicalEO}, and is suitable for infections that provide long-lasting immunity, such as measles or smallpox.
For diseases that do not confer lasting immunity and therefore allow for repeated infections, such as gonorrhea and other sexually transmitted infections, the Susceptible–Infected–Susceptible (SIS) model is more appropriate. This modeling framework has been widely studied in the epidemiological literature, particularly in the context of sexually transmitted diseases \cite{cooke1973some,hethcote1973asymptotic,GarnettAnderson1996}. In the SIS model, individuals who recover from infection immediately return to the susceptible class, as no immunity is acquired. On the other hand, several infectious diseases confer only temporary immunity. In such cases, individuals may recover, lose immunity after a finite period, and become susceptible again, as observed for diseases such as influenza, whooping cough, and certain coronaviruses. To capture this epidemiological feature, the Susceptible–Infected–Recovered–Susceptible (SIRS) model introduces a feedback transition from the recovered compartment to the susceptible compartment, allowing immunity to wane over time. The SIRS model has twice the dimension of the SIS model and the same as the SIR model, where however the fractions of susceptible and recovered individuals are monotone  in time, implying convergence to an equilibrium point from every initial state.  This significantly increases the analytical complexity of the SIRS model.
Early studies include~\cite{liu1986influence} analyzing the influence of nonlinear incidence functions on stability and persistence, while further nonlinear effects, including bifurcation phenomena and complex dynamics, were explored in~\cite{Jin2007}.  A key aspect concerns the global stability of equilibria. In this context, 
\cite{KOROBEINIKOV2002955} discussed suitable Lyapunov functions for single population SIRS models and conditions for their asymptotic global stability. 

Classical compartmental epidemic models are based on the assumption of homogeneous mixing, which often fails to represent the heterogeneous contact patterns observed in real populations. Graph theory provides a natural mathematical framework for describing such structured interactions. 
Network epidemic models have been considered both at the microscopic level ---where nodes represent individuals--- and at the macroscopic level ---where nodes represent subpopulations or homogeneous groups of individuals. A fundamental contribution within the latter class is the study of the network $SIS$ model  in \cite{lajmanovich1976deterministic}, where, exploiting the fact that this is a monotone dynamical system \cite{HirschMonotoneSystems06}, the authors identify the dominant eigenvalue of the interaction matrix as the key parameter determining, along with the recovery rate, the threshold for the existence and stability of an endemic equilibrium. 
Graph theory has also been applied to the SIR epidemic model to account for heterogeneous contact structures~\cite{pastor2015epidemic,pare2020modeling,Alutto.ea:2025}. More recent contributions, including the reviews \cite{Nowzari.ea:2016,Mei.ea:2017,Brauer2017,ZinoCao2021}, provide a comprehensive overview of the evolution from classical scalar compartmental models to network-based approaches. 


More recent contributions have extended the SIRS framework to structured and network-based populations. Epidemic spreading in heterogeneous complex networks was analyzed in~\cite{Li2014} and \cite{chen2014global}, while layered and multiplex network formulations were proposed in~\cite{zhang2021layered} to model interacting spreading processes. Threshold phenomena and epidemic persistence in network-based SIRS models were investigated in~\cite{saif2019epidemic}, and global stability properties under nonmonotone incidence rates were established in~\cite{lijun-nonmonotone-2018}. Importantly, in these works, the type of interactions considered can be effectively recast as a rank-one matrix, i.e., a matrix that can be written as the outer product of two vectors. This corresponds to assuming that the interaction between two nodes factorizes into node-specific properties (e.g., susceptibility and infectivity), thus neglecting pairwise-specific effects, such as homophily. While this representation simplifies the analysis, it imposes a strong structural limitation, as it only captures a single mode of interaction and cannot account for more complex or heterogeneous connectivity patterns. 
In \cite{yang2017heterogeneous}, the stability properties of a heterogeneous node-based SIRS model on a finite, undirected network are investigated. The model accounts for node-dependent epidemiological parameters and network heterogeneity, and the authors derive conditions ensuring the existence, uniqueness, and stability of the endemic equilibrium. In particular, the authors establish a global stability result under additional sufficient conditions involving the network and model parameters. In \cite{MUROYA20131693}, a multi-group SIRS epidemic model with varying population sizes and cross-group infection is investigated, and the global stability of the endemic equilibrium is established under some sufficient conditions.

In this paper, we analyze the behavior of a network SIRS epidemic model over a finite connected graph. The interaction matrix $W$ is assumed to be nonnegative and irreducible, with no further structural assumptions, allowing for non-separable interactions and providing a more flexible and expressive modeling framework. The reproduction number $R_0$ is defined as the dominant eigenvalue of a rescaled version of the interaction matrix $W$, whose rows are normalized by the corresponding recovery rates $\gamma_i$. We first study existence and uniqueness of equilibrium points. We reduce the problem of finding equilibria of the system to that of finding a fixed point of a monotone function in a space of half the dimension. In the supercritical regime $R_0>1$, we prove the existence and uniqueness of the endemic equilibrium. As a byproduct of our proof, we obtain a distributed iterative algorithm for computing the endemic equilibrium, which represents a significant advantage both from a theoretical and a computational perspective. Furthermore, we present a special case in which an explicit expression for the endemic equilibrium can be derived. Then, we study the asymptotic stability of the equilibrium points, when they exist. If $R_0\le1$, the disease-free equilibrium is globally asymptotically stable; if $R_0>1$, it is unstable. The main contribution of this paper is to prove asymptotic stability of the endemic equilibrium above the epidemic threshold. What makes this challenging 
is that the network SIRS epidemic model is not monotone, unlike the network SIS model \cite{lajmanovich1976deterministic}; the susceptible and recovered fractions are not monotone in time, unlike in the network SIR model \cite{Mei.ea:2017,Alutto.ea:2025}; and no global Lyapunov function is known except for the case of the network SIRS model with rank-one interaction matrix and homogeneous recovery rate, cf.~\cite{chen2014global} and Appendix \ref{sec:rank-one}. Moreover, our analysis is complementary of \cite{yang2017heterogeneous}, where the global stability of the endemic equilibrium is investigated: under fewer structural and parameter assumptions, we establish local asymptotic stability of the endemic equilibrium. A similar remark applies to \cite{MUROYA20131693}, where the global stability of the endemic equilibrium is established for a multi-group SIRS model under a sufficient condition that relies on the presence of vital dynamics, and is therefore not applicable to our setting. We develop an ad hoc argument based on the Schur complement and the Gershgorin theorem, which allows us to establish the local asymptotic stability of the endemic equilibrium. Finally, numerical simulations are provided to illustrate the theoretical findings and to further investigate the dynamic behavior of the system. In addition to confirming the analytical results, the simulations suggest new conjectures concerning the region of asymptotic stability of the endemic equilibrium. 

The paper is organized as follows. In \Cref{sec:model}, we introduce the network SIRS epidemic model and establish its well-posedness. In \Cref{sec:equilibria}, we discuss the existence and uniqueness of equilibrium points. Our main stability result is presented in \Cref{sec:stability}, where we analyze the equilibrium points both below and above the epidemic threshold. Finally, in \Cref{sec:simulations}, we report numerical results that support the theoretical analysis and motivate additional conjectures.

\section{Network SIRS epidemic model}\label{sec:model}
In this section, we introduce the network SIRS epidemic model and present some preliminary results. 

The interaction network is modeled as a directed, weighted graph $\mathcal{G} = (\mathcal{V}, \mathcal{E}, W)$, where $\mathcal V$ is a nonempty finite set of nodes representing populations of individuals, $\mathcal{E}\subseteq \mathcal V \times \mathcal V$ is the set of directed links, and $W$ in $\R_+^{\mc V\times\mc V}$ is the interaction matrix whose entries satisfy $W_{ij}>0$ if and only if $(i,j)\in\mc E$. 
For every population $i$ in $\mc V$, the time-dependent variables $x_i(t)$, $y_i(t)$, and $z_i(t)$ denote the fractions of susceptible, infected, and recovered individuals, respectively. Infection spreads through contacts among populations, whose frequency is described by the interaction matrix $W$. Our analysis is restricted to connected graphs, in which there exists a path from any node to every other node; in matrix terms, this corresponds to the interaction matrix $W$ being irreducible. Connectivity ensures that infection can spread from any population to all others, either directly or indirectly through the network of interactions.
Recovery and loss of immunity for each population $i$ in $\mc V$ occur at rates $\gamma_i>0$ and $\delta_i>0$, respectively, which are assumed to be constant throughout the paper.

The network SIRS epidemic model is then described by the following system of ordinary differential equations: 
\be\label{sirs_network_sistema}	
	\begin{cases}
		\dot{x}_i(t)= - \sum_{j} W_{ij}\, y_j(t)\, x_i(t) + \delta_i z_i(t)\\[6pt]
		\dot{y}_i(t) = \sum_{j} W_{ij}\, y_j(t)\, x_i(t) - \gamma_i y_i(t)\\[6pt]
		\dot{z}_i(t)= \gamma_i y_i(t) - \delta_i z_i(t)
	\end{cases}
\ee
for every $i$ in $\mathcal{V}$. 
Equation \eqref{sirs_network_sistema} defines a dynamical system with state space 
	\be\label{space_state_def}
		\mathcal{X}=\{(x,y,z)\in \mathbb{R}_+^{\mathcal V}\times\mathbb{R}_+^{\mathcal V}\times \mathbb{R}_+^{\mathcal V}: x+y+z=\mathbf{1}\}
	\ee
as formalized in the following standard result. 
\begin{lemma}\label{lemma:well_pos}
	Consider the network SIRS epidemic model  \eqref{sirs_network_sistema} with interaction matrix $W$, recovery rate profile $\gamma$, and loss-of-immunity rate profile $\delta$. For every initial condition $(x(0),y(0),z(0))$ in $\mathcal{X}$, there exists a unique solution $(x(t), y(t), z(t))$  in $\mathcal{X}$ for all $t \ge 0$.
\end{lemma}

\begin{proof}
	See Appendix \ref{app:well_pos}.
\end{proof}
Since $x_i+y_i+z_i=1$ for all $i$, we can equivalently study the reduced system 
	\be\label{sistema_ridotto_network}
		\begin{cases}
			\dot y_i = \displaystyle\sum_{j} W_{ij}\,y_j\,(1-y_i-z_i)-\gamma_i y_i\\[6pt]
			\dot z_i = \gamma_i y_i - \delta_i z_i\,,
		\end{cases}
	\ee
which can be written in the compact form
	\be\label{sistema_ridotto_vettoriale}
		(\dot{y}, \dot{z})=f(y,z)=
		\begin{pmatrix}
			[1-y-z]Wy-[\gamma] y\\
			\ga y-\de z
		\end{pmatrix}
	\ee

From an epidemiological viewpoint, once the infection is introduced in the network, it cannot disappear instantaneously. The following result formalizes this intuition by proving that, if the infected fraction is initially strictly positive in at least one node, then the infection remains strictly positive in all nodes at every positive time. 
\begin{proposition}\label{theo:system_pos}
	Consider the network SIRS epidemic model  \eqref{sirs_network_sistema} with irreducible interaction matrix $W$ and positive recovery and loss-of-immunity rate profiles $\gamma$ and $\delta$. Then, for every initial condition $(x(0),y(0),z(0))$ in $ \mathcal X$:
\begin{enumerate}	
\item[(i)] if $y(0)=\mathbf0$, then $y(t)=\mathbf 0$ for every $t\ge0$;
\item[(ii)] if $y(0)\gneq\mathbf0$, then $y(t)>\mathbf 0$ for every $t>0$. 
\end{enumerate}
\end{proposition}
\begin{proof}
	See Appendix \ref{app:lemma_positiveness}. 
\end{proof}

\section{Equilibrium points}\label{sec:equilibria}
In this section, we study the existence and uniqueness of equilibrium points of the network SIRS epidemic model. 
Clearly, $(\mathbf 1,\mathbf 0,\mathbf 0)$ is always an equilibrium point of \eqref{sirs_network_sistema}; we refer to it as the disease-free equilibrium. Observe that, if $(x^*,y^*,z^*)$ with $y^*\gneq\0$ is another equilibrium point, then, by Proposition \ref{theo:system_pos}(ii), it must be $y^*>\0$. We refer to any such equilibrium point  as an endemic equilibrium. 

For every $\alpha $ in $ \mathbb R_+^{\mathcal{V}}$, let $\ov y$ in $\R^{\mc V}$ be the vector with entries 
$\ov y_i=(1+\alpha_i)^{-1}$ for every $i$ in $\mc V$, and let 
$\Psi:\mathbb R_+^{\mathcal{V}}\times \mathbb R_+^{\mathcal{V}}\to \mathbb R_+^{\mathcal{V}}$ be defined by
\be\label{def_Psi}
	(\Psi(y,\alpha))_i=\dfrac{y_i}{1+(1+\alpha_i)y_i}\,,\qquad \forall i\in\mc V\,.
\ee
Then, consider the set
$\mathcal Y_{\alpha} = \{y\in\R^{\mc V}:\0\le y\le \ov y\}$ 
and define 
$\Phi: \mc Y_{\alpha}\times \mathbb R_+^{\mathcal{V}\times \mathcal V}\times \mathbb R_+^{\mathcal{V}}\to \mc Y_{\alpha}$ by
\be\label{def_Phi}
	\Phi(y,M,\alpha)=\Psi(My,\alpha).
\ee
The next result gathers some properties of the function $\Phi$.
\begin{lemma}\label{lem:H}
	Consider the function $\Phi$ defined in \eqref{def_Phi}. Then:
	\begin{enumerate}
		\item[(i)] \label{lem:H-i}$\Phi(\mathbf0,M,\alpha)=\mathbf0$, for every $M$ and $\alpha$; 
		\item[(ii)] \label{lem:H-iii}$\Phi(y,M,\alpha)<My$, for every $y>\0$ and irreducible $M$; 
		\item[(iii)] $\Phi(y,M,\alpha)$ is nondecreasing in $y$ and $M$, and nonincreasing in $\alpha$.
	\end{enumerate}
\end{lemma}
\begin{proof}
	See Appendix \ref{app:lem_H}.
\end{proof}\medskip

The following result reduces the problem of finding equilibrium points of the network SIRS epidemic model \eqref{sirs_network_sistema} to that of finding fixed points of $\Phi(\,\cdot\,,\ga^\inv W,\de^\inv \gamma)$ on $\mc Y_{\alpha}$.
\begin{lemma}\label{equivalence_equilibrium_fixed}
	Consider the network SIRS epidemic model \eqref{sirs_network_sistema} with irreducible interaction matrix $W$ and positive recovery and loss-of-immunity rate profiles $\gamma$ and $\delta$.  Then, $(x^*,y^*,z^*)$ in $\mathcal X$ is an equilibrium point if and only if $$
	y^*=\Phi(y^*,\ga^\inv W,\de^\inv \gamma)\,,\qquad  z^*=\de^\inv \gamma y^*\,,$$and $x^*=\1-y^*-z^*$.
\end{lemma}
\begin{proof}
	See Appendix \ref{app:proof_equivalence}.
\end{proof}\medskip

	\begin{theorem}\label{theorem_fixed_point}
	Let $M$ in $\mathbb R_+^{\mathcal{V}\times\mathcal{V}}$ be irreducible and let $\lambda_{M}=\rho(M)$ be its dominant eigenvalue. 
	For $\alpha$ in $\mathbb R_{++}^{\mathcal{V}}$ and  $\xi_0$ in $\mc Y_{\alpha}$, let $(\xi_0,\xi_1,\ldots)\subseteq\mc Y_{\alpha}$ be generated by the recursion 
	\be\label{definition_discrete_time}\xi_{k+1}=\Phi(\xi_k,M,\alpha)\,, \qquad k=0,1,\ldots\,.\ee
	Then:
	\begin{enumerate}
		\item[(i)] \label{fixed_point1} if $\lambda_{M}\le 1$, $\mathbf0$ is the unique fixed point of $\Phi$ in $\mc Y_{\alpha}$ and for every $\xi_0$ in $ \mathcal Y_\alpha$, $\lim_{k\to\infty} \xi_k=\0$;
		\item[(ii)]\label{fixed_point2} if $\lambda_{M}>1$, $\Phi$ admits two distinct fixed points  in $\mc Y_{\alpha}$: $\mathbf{0}$ and $y^*>\0$, and 
		$\lim_{k\to\infty} \xi_k=y^*$ for every $\xi_0>\0$.
	\end{enumerate}
\end{theorem}

\begin{proof}
Let $\{\ov \xi_k\}_k$ be the sequence generated by the recursion \eqref{definition_discrete_time} with initial condition $\ov \xi_0=\ov y$. Since $\ov y$ is the maximal element of $\mathcal Y_{\alpha}$,  
we have that $\bar \xi_1\le \bar y = \bar \xi_0$, and since $\Phi$ is nondecreasing in $y$ by \Cref{lem:H}(iii), we get that 
			$$\bar \xi_{k+1}=\Phi^k(\bar \xi_1,M,\alpha)\le \Phi^k(\bar \xi_0,M,\alpha)=\bar \xi_k\,,$$
			for every $k\ge0$, i.e., the sequence $(\bar \xi_k)$ is nonincreasing and, since $\mathcal Y_{\alpha}$ is compact, it must converge to some limit $\bar \xi^*=\lim_{k\to\infty}\bar \xi_k$. 
			Continuity of $\Phi$ implies that 
$$\bar \xi^*\!=\!\lim_{k\to\infty}\bar \xi_k\!=\!\lim_{k\to\infty}\bar \xi_{k+1}\!=\! \lim_{k\to\infty}\Phi(\bar \xi_k,M,\alpha)\!=\!\Phi(\bar \xi^*,M,\alpha)\,,$$
			i.e., $\bar \xi^*$ is a fixed point of $\Phi(\,\cdot\,,M,\alpha)$. Moreover, monotonicity of $\Phi(y,M,\alpha)$ in $y$ (Lemma \ref{lem:H-iii}(iii)) implies that $y^*\le\bar \xi^*$ for every other fixed point $y^*=\Phi(y^*,M,\alpha)$ in $\mc Y_{\alpha}$.

(i) From \Cref{lem:H}(i), $\0=\Phi(\0,M,\alpha)$ is a fixed point for every $M$ and $\alpha$. 
If $\bar \xi^*>\0$, then \Cref{lem:H}(ii) implies that $\bar \xi^*=\Phi(\bar \xi^*,M,\alpha)<M\bar \xi^*$, so that  
			$\min_i{(M\bar \xi^*)_i}/{\bar \xi^*_i}>1$, and the max--min characterization of the dominant eigenvalue of nonnegative irreducible matrices \cite[Chapter 2.2, Equation (2.11)]{berman1994nonnegative} implies that
		\be\lambda_M=\max_{x>0} \min_{i\in\mathcal V}{(Mx)_i}/{x_i}\ge\min_{i\in\mc V} {(My)_i}/{y_i}>1\,.\ee
Therefore, if $\lambda_{M}\le1$, then $\bar\xi^*=\mathbf0$ is the unique fixed point of $\Phi(\,\cdot\,,M,\alpha)$. Moreover, for every $\xi_0$ in $ \mathcal Y_\alpha$, 
			$\mathbf 0 \le \xi_0 \le \bar y$ and Lemma \ref{lem:H-iii}(iii) imply that 
$$\mathbf 0 = \Phi^k (\0,M,\alpha)\le \Phi^k(\xi_0,M,\alpha)\le \Phi^k(\bar y,M,\alpha)=\bar \xi_{k+1}\,,$$
for every $k\ge0$, so that $\mathbf 0 \le \lim_{k} \xi_k\le \lim_{k} \bar \xi_k=\bar \xi^*=\0\,.$
Hence, $\lim_{k} \xi_k=\mathbf 0 $ for every $\xi_0$ in $\mc Y_{\alpha}$, thus completing the proof of part (i). 
			
(ii) Let $v>\0$ be a right eigenvector of the irreducible nonnegative matrix $M$ associated to $\lambda_M$. 
		For $\varepsilon>0$, $\Phi(\varepsilon v,M,\alpha)=\Psi\!\left(M\varepsilon v,\alpha\right)=\Psi\!\left(\lambda_{M}\varepsilon v,\alpha\right)$. 
		Since $\Phi$ is differentiable, $\Phi(\0,M,\alpha)
		\!=\!\0$, and 
			$\nabla \Phi(\0)\!=\!\nabla \Psi(\0)M\!=\!M$,  
		$$\Phi(\varepsilon v,M,\alpha)=M\varepsilon v+o(\varepsilon)
				=\lambda_{M} \varepsilon v + o(\varepsilon)\,,\quad\text{as }\eps\to0\,.$$ 
If $\lambda_{M}>1$, then there exists $ \varepsilon^*>0$ such that 
		\be\label{Phi(y0)_gey0}\Phi(\varepsilon^*v)\ge\varepsilon^*v\,.\ee 
		Let $\ul \xi_0=\varepsilon^*v$ and $\ul \xi_k$, for $k\ge1$, be generated by the recursion \eqref{definition_discrete_time}. 
			Then, from \eqref{Phi(y0)_gey0}, $\ul \xi_1=\Phi(\ul \xi_0)\ge \ul \xi_0$, and monotonicity of $\Phi$ implies that $(\ul \xi_k)$ is a nondecreasing bounded sequence. Arguing as done earlier in the proof for $\bar\xi_k$, $\ul\xi_k$ necessarily converges to some positive fixed point $\ul \xi^*=\Phi(\ul \xi^*,M,\alpha)>\0$. Now, let $\delta=\1$, $\gamma=\alpha$, and $W=[\gamma]M$, so that $[\gamma]^{-1}W=M$ and $[\delta]^{-1}\gamma=\alpha$. Notice that $W$ is irreducible since so is $M$ and $\alpha$ is positive. Then, \Cref{equivalence_equilibrium_fixed} implies that $\ul \xi^*$ and $\ov \xi^*$ are equilibrium points of the network SIRS epidemic model \eqref{sirs_network_sistema}.
		Assume by contradiction that $\ul \xi^*\lneq \ov \xi^*$. 
		Without loss of generality, assume that 
		\be\label{eq5+6}
			\ul \xi^*_1<\ov \xi^*_1\,,\qquad 
{\ov \xi_j^*}/{\ul \xi_j^{*}}\le{\ov \xi^*_1}/{\ul \xi_1^{*}} , \qquad \forall j\in\mathcal V .\ee
		Then, we have
		\begin{align*}
			0 &= \frac{\ul \xi_1^{*}}{\ov \xi^*_1} \left( (1-(1+\frac{\gamma_1}{\delta_1})\ov \xi^*_1) \sum\nolimits_j W_{1j} \ov \xi_j^* - \gamma_1 \ov \xi_1^* \right) \\[0pt]
			&<  (1-(1+\frac{\gamma_1}{\delta_1})\ul \xi_1^{*}) \sum\nolimits_j W_{1j} \frac{\ul \xi_1^{*}}{\ov \xi_1^*} \ov \xi_j^* - \gamma_1 \ul \xi_1^{*} \\[0pt]
			&\leq  (1-(1+\frac{\gamma_1}{\delta_1})\ul \xi_1^{*}) \sum\nolimits_j W_{1j}\ul \xi_j^{*} - \gamma_1 \ul \xi_1^{*} \\[6pt]
			&= 0 ,
		\end{align*}
		where the first and last equalities follow from the fact that both $\ul \xi^*$ and $\ov \xi^{*}$ are equilibrium points, and the inequalities follow  from \eqref{eq5+6}. As we have reached a contradiction, we can conclude that $\ul \xi^* =\ov \xi^{*}=y^*$, hence there exists a unique fixed point $y^*$ in $\mc Y_\alpha$ of $\Phi$ such that $y^*>\0$.
		
		For every $\xi_0>\0$,  let $\xi_k$ for $k\ge1$ be generated by the recursion \eqref{definition_discrete_time}. Then there exists $\varepsilon>0$ such that $\0<\varepsilon v\le \xi_0$ and $\Phi(\varepsilon v)>\varepsilon v$. Hence, we have
$\varepsilon v\le \xi_0 \le \ov y$.
		Consider the nondecreasing sequence $\{\ul \xi_k\}_k$ defined above with initial condition $\ul \xi_0=\varepsilon v$ and the nonincreasing one $\{\ov \xi_k\}_k$. Hence,
			\[
			y^*=\lim_{k\to\infty}\Phi^k(\ul \xi_0)\le \lim_{k\to\infty}\Phi^k(\xi_0) \le \lim_{k\to\infty}\Phi^k(\ov \xi_0) =y^*
			\]
			Hence, $\lim_{k} \xi_k=y^*>\0 $ for every $\xi_0>\0$, thus completing the proof of part (ii). 
\end{proof}\medskip

Theorem \ref{theorem_fixed_point} suggests defining the basic reproduction number as
\be R_0=\rho(\ga^\inv W)=\lambda_{\ga^\inv W}
\,,\ee
which yields the following result. 
\begin{corollary}\label{corollario_sirs_equilibrium}
	Consider the network SIRS epidemic model  \eqref{sirs_network_sistema} with irreducible interaction matrix $W$ in $\mathbb R_+^{\mathcal V\times \mathcal V}$ and positive recovery and loss-of-immunity rate vector $\gamma$ and $\delta$ in $ \mathbb R_{++}^{\mathcal{V}}$. 
	 Then:
	\begin{enumerate}
		\item[(i)] \label{corollario_1}if $R_0\le 1$, the disease-free equilibrium  $(\1,\0,\0)$ is the unique equilibrium point;
		\item[(ii)] \label{corollario_2}if $R_0>1$, the system admits the disease-free equilibrium $(\1,\0,\0)$ and the endemic equilibrium point $(\1-y^*-z^*,y^*,z^*)$, with $y^*>\mathbf0$ and $z^*=\de^\inv \gamma y^*$;
		\item[(iii)] \label{corollario_3}$y^*$ is a nondecreasing function of $\ga^\inv W$ and nonincreasing function of $\de^\inv \gamma$.
	\end{enumerate}
\end{corollary}
\begin{proof}
	 Using \Cref{equivalence_equilibrium_fixed}, (i),(ii) follow directly from \Cref{theorem_fixed_point} with $M=\ga^\inv W$ and $\alpha=\de^\inv \gamma$. 
	Monotonicity of $y^*$ in $M=\ga^\inv W$ and $\alpha=\de^\inv \gamma$ follow from \Cref{lem:H}(iii).
\end{proof}\medskip

\begin{remark}
	\Cref{theorem_fixed_point} implies that, if $R_0>1$, then the endemic equilibrium point $y^*>\0$ is the limit of the sequence generated by recursion \be\label{algo}\xi_{k+1}=\Phi(\xi_k,\ga^{-1}W,\de^\inv \gamma)\,,\qquad k\ge0\,,\ee initiated with any $\xi_0>\0$ in $\mc Y_{\alpha}$. Recursion \eqref{algo} effectively provides an iterative distributed algorithm for the computation of the endemic equilibrium point of the network SIRS epidemic model \eqref{sirs_network_sistema}. However, convergence of this algorithm does not have any direct implication on the stability of the endemic equilibrium of the network SIRS epidemic model \eqref{sirs_network_sistema}.
	In fact, observe that \eqref{algo} defines a discrete-time monotone $n$-dimensional dynamical system, while \eqref{sirs_network_sistema} is a continuous-time non-monotone $2n$-dimensional dynamical system. The stability analysis of the network SIRS epidemic model \eqref{sirs_network_sistema} will be the object of the next section. 
\end{remark}

\begin{example}\label{ex:out_regular_example}
	We now consider an out-regular graph, which makes it possible to compute the endemic equilibrium. Consider the interaction matrix such that $W\1=\lambda_W \1$ and homogeneous recovery and loss-of-immunity rates $\gamma=\bar \gamma \1$ and $\delta=\bar \delta \1$, for some $\bar\gamma,\bar\delta>0$. In this case, if $R_0=\lambda_W/\bar\gamma>1$, the endemic equilibrium is $y^*=\frac{\bar\delta}{\bar \gamma +\bar\delta}(1-\frac{\bar\gamma}{ \lambda_W})1$. Indeed, since $Wy^*=\lambda_W y^*$, the equilibrium equation \eqref{equilibrium_eq_y} can be reduced as 
	 \[\bigg[1-\frac{\bar \gamma +\bar\delta}{\bar\delta}y^*\bigg]R_0 y^*= y^*\]
	 Solving the equation componentwise yields $y_i^*=\frac{\bar\delta}{\bar \gamma+\bar\delta}(1-\frac{1}{R_0})$ for every $i$ and, from the uniqueness proved in \Cref{theorem_fixed_point}, the out-regular case with homogeneous rates allows the solution to be found explicitly. Having an out-regular graph means all nodes have the same out-degree, so each population has equal exposure to infection, but not necessarily equal potential to spread it. 
\end{example}

\section{Stability Analysis}\label{sec:stability}
In this section we analyze the stability of the network SIRS epidemic model \eqref{sirs_network_sistema}. We manage to derive the global asymptotic stability of the disease-free equilibrium below the threshold $R_0=1$ and, more importantly, the local asymptotic stability of the endemic equilibrium point above the threshold. We also investigate the region of attraction of the disease-free equilibrium when the endemic equilibrium exists. 
\subsection{Stability of the disease-free equilibrium point}
For an irreducible interaction matrix $W$ and positive recovery rate profile $\gamma$, let $\bar{v}>\0$ be the leading left eigenvector of $\ga^\inv W$ and define the  function:
\be\label{lyap_fun}
	V=\bar{v}^T\ga^\inv y\,.
\ee

\begin{lemma}\label{lemma:lyap_fun}
	Consider the network SIRS epidemic model \eqref{sirs_network_sistema} with irreducible interaction matrix $W$ and positive recovery and loss-of-immunity rate profiles $\gamma$ and $\delta$. Then  \be\label{derivative_V}\dot{V}=(R_0-1)\bar{v}^T y-\bar{v}^T\ga^\inv [y+z]Wy\,.\ee 
Moreover, if $R_0\le 1$, then $\dot{V}\le0$ 
	and the largest invariant subset of $\{\dot{V}=0\}$ is $\{ y=\0\}$.
\end{lemma}
\begin{proof}
	See Appendix \ref{app:lyap_fun}. 
\end{proof}\medskip

The next theorem establishes the stability properties of the network SIRS epidemic model below the epidemic threshold. 
\begin{theorem}\label{theo:stability-DF}
	Consider the network SIRS epidemic model  \eqref{sirs_network_sistema} with irreducible interaction matrix $W$ in $\mathbb R_+^{\mathcal V\times \mathcal V}$ and positive recovery and loss-of-immunity rate vectors $\gamma$ and $\delta$ in $ \mathbb R_{++}^{\mathcal{V}}$. Then, if $R_0\le1$, the disease-free equilibrium $(\1,\0,\0)$ is globally asymptotically stable.  
\end{theorem} 
\begin{proof}
Let $V$ be defined as in \eqref{lyap_fun}. Since both $\bar{v}$ and $\gamma$ are positive, $V$ is positive definite with respect to $y=\mathbf 0$. 
Then, \Cref{lemma:lyap_fun} and LaSalle’s invariance principle imply that 
	$y(t)\to\0$ as $t\to\infty$. 
	Hence, for every $i$ in $\mc V$, we have
	\[z_i(t)=z_i(0)e^{-\delta_i t}+\gamma_i \int_0^t e^{-\delta_i(t-s)}y_i(s)ds\stackrel{t\to\infty}{\longrightarrow} 0\,, \]
	 proving that $(\1,\0,\0)$ is globally asymptotically stable. \end{proof}
	
\subsection{Stability of the endemic equilibrium point}
We now investigate the asymptotic stability of the endemic equilibrium when $R_0>1$, starting with the following preliminary result. 
\begin{lemma}\label{lemma:S_invertible}
	Consider the network SIRS epidemic model \eqref{sirs_network_sistema} with irreducible interaction matrix $W$ and positive recovery and loss-of-immunity rate profiles $\gamma$ and $\delta$. Assume that $R_0>1$ and let $(x^*,y^*,z^*)$ in $\mathcal X$ be the endemic equilibrium of the system. Define \be\label{eta-def}\eta=\min_{i\in\mc V} \min\{(Wy^*)_i,\delta_i\}\,,\ee and, for $\lambda $ in $ \mathbb{C}$, consider the matrix
$$S(\lambda) = [x^*] W - [W y^*] - \ga - \lambda I - \ga(\de+\lambda I)^\inv [W y^*].$$
Then, $S(\lambda)$ is invertible for every $\lambda$ such that $\Re(\lambda) >-\eta$.
\end{lemma}
\begin{proof}
	See Appendix \ref{app:lemma_schur}.
\end{proof}\medskip

The following theorem establishes the stability behavior of an SIRS network model above the threshold $R_0=1$.
\begin{theorem}\label{theo:stability-END}
	Consider the network SIRS epidemic model  \eqref{sirs_network_sistema} with irreducible interaction matrix $W$ and positive recovery and loss-of-immunity rate profiles $\gamma$ and $\delta$. If $R_0>1$, then:
	\begin{enumerate} 
	\item[(i)] the disease-free equilibrium $(\mathbf 1,\mathbf 0,\mathbf 0)$ is unstable and its region of attraction is $\{y=\0\}$;
		\item[(ii)]  the endemic equilibrium $(x^*,y^*,z^*)$ is locally asymptotically stable.  
		\end{enumerate}
\end{theorem}
\begin{proof}
(i) The Jacobian matrix of the reduced system \eqref{sistema_ridotto_vettoriale} in the disease-free equilibrium is	
$$\nabla f(\0,\0)= \begin{pmatrix}
			W-\ga & 0\\
			\ga & -\de 
		\end{pmatrix}\,.$$
If $R_0=\rho(\ga^{-1}W)>1$, then $W-\ga$ is unstable, hence, so is the disease-free equilibrium. Observe that, for every initial state $(x(0),y(0),z(0))$ with $y(0)=\0$, we have $\dot{y}=\0$ and $\dot{z}=-\de z$, so that $z(t)=e^{-[\delta]t}z(0)\stackrel{t\to+\infty}{\longrightarrow}0$. Hence, the region of attraction of the disease-free equilibrium contains the set $\{y=\0\}$. On the other hand, consider an initial state $(x(0),y(0),z(0))$ such that $y(0)\gneq\0$. By Proposition \ref{theo:system_pos}(ii),  $y(t)>\0$ for $t>0$. Arguing by contradiction, assume that $(x(t),y(t),z(t))\stackrel{t\to+\infty}{\longrightarrow}(\1,\0,\0)$.
Then, for every $0<\eps<(R_0-1)/R_0$, there exists $T>0$ such that $y_i(t)+z_i(t)\le \eps$, for every $i$ in $\mc V$ and  $t \ge T$, so that, by Lemma \ref{lemma:lyap_fun},  the time derivative of the function $V$ defined in \eqref{lyap_fun} satisfies  
	\begin{align*}
		\dot{V}&=(R_0-1)\bar{v}^T y-\bar{v}^T\ga^\inv[y+z]Wy\\
		&\ge(R_0 -1 ) \bar{v}^T y-\eps \bar{v}^T\ga^\inv Wy\\
		&=(R_0-1-\eps R_0)\bar{v}^T y>0\,,
	\end{align*}
	so that $V$ cannot converge to $0$ as $t\to\infty$. Since $V$ is positive definite with respect to $\{y=\0\}$ and continuous, then $y$ cannot converge to $\0$, which leads to a contradiction. 
Hence, if $R_0>1$, the region of attraction of $(\1,\0,\0)$  is exactly the set $\{y=\0\}$. 

(ii) The Jacobian matrix of \eqref{sistema_ridotto_vettoriale} in the endemic equilibrium is  
$$\nabla f(y^*,z^*)= \begin{pmatrix}
			[x^*]W-[Wy^*]-\ga & -[Wy^*]\\
			\ga & -\de 
		\end{pmatrix}.$$
		Let $\eta$ be defined as in \eqref{eta-def}. 
	For every $\lambda$ in $\C$ such that $\Re(\lambda)>-\eta $, the diagonal matrix $-(\de+\lambda I)$ is invertible, so that, using the Schur complement
	\be\label{schur_complement}
		S(\lambda)=[x^*]W-[Wy^*]-\ga-\lambda I-\ga (\de +\lambda I)^\inv[Wy^*]
	\ee
 we have that 
$$\ba{rcl}\operatorname{det}(\nabla f(y^*,z^*)-\lambda I )
		&=&\operatorname{det}(-(\de+\lambda I)) \operatorname{det}(S(\lambda)) \\[5pt]
		&=&(-1)^n\prod_i(\delta_i+\lambda)\operatorname{det}(S(\lambda))\,.\ea$$
	Since $\operatorname{det}(S(\lambda))\ne0$ by Lemma \ref{lemma:S_invertible}, we have that 
	$\operatorname{det}(\nabla f(y^*,z^*)-\lambda I )\ne0$ for every $\lambda$ such that $\Re(\lambda)>-\eta$, i.e., 
$$\Re(\lambda)\le-\eta<0\,,\qquad \forall\lambda\in\sigma(\nabla f(y^*,z^*))\,.$$
This proves that $\nabla f(y^*,z^*)$ is Hurwitz stable, so that, by the linearization theorem, the endemic equilibrium $(x^*,y^*,z^*)$ of the network SIRS epidemic model is locally asymptotically stable. 
\end{proof}
Theorem \ref{theo:stability-END} does not characterize the region of attraction of the endemic equilibrium, as it only shows that it is a nonempty open subset of $\mathcal X\setminus \{y=\0\}$. However, numerical simulations, such as those reported in \Cref{sec:simulations} lead one to conjecture that  the region of attraction of the endemic equilibrium is indeed the whole $\mathcal X\setminus \{y=\0\}$. In fact, this is known to hold true in the single population case and more in general in the special case of rank-one interaction matrix and homogeneous recovery rate profile, as shown in Appendix \ref{sec:rank-one}.

	\section{Numerical simulations}\label{sec:simulations}
	In this section, we consider a connected network of $n=5$ populations. The loss-of-immunity rate is fixed as
	$
	\delta=(
		0.3, 0.4, 0.2,0.1,0.6)$
	and 
	\[
	\ga^\inv W=\begin{bmatrix}
		3.0000 & 6.0000 & 4.0000 & 1.0000 & 8.0000 \\
		0.1000 & 0.4000 & 1.0000 & 0      & 0.5000 \\
		2.0000 & 1.4000 & 2.8000 & 2.0000 & 1.4000 \\
		0.6000 & 0      & 0      & 1.2000 & 0.4000 \\
		2.8571 & 0      & 0      & 0.7143 & 1.2857
	\end{bmatrix}\]
	with $R_0=8.7434>1$. In \Cref{fig:infected_rec_2}, we perform $8$ different simulations and highlight the behavior of a single population in the network. We observe that, for all the simulated initial conditions in $\mathcal X\setminus\{y=\mathbf{0}\}$, the trajectories converge to the endemic equilibrium, while trajectories that start in $\{y=\0\}$ are attracted to the disease-free equilibrium, in accordance with Theorem \ref{theo:stability-END}(i). Our conjecture is that the region of attraction of the endemic equilibrium coincides with $\X\setminus\{y=\0\}$ and these numerical results provide strong evidence in support of this claim. 
	\begin{figure}
		\centering 
		\includegraphics[width=4cm]{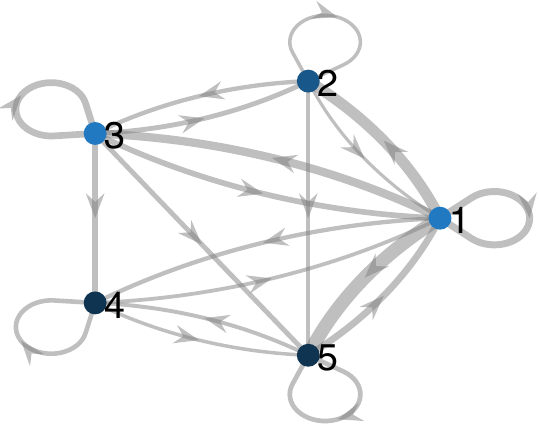}
		\caption{Interaction network  in Section \ref{sec:simulations}. Each link's thickness is proportional to the corresponding entry  $\ga^{-1}W$. }
	\end{figure}
\begin{figure}
	\centering
	\includegraphics[width=0.35\textwidth]{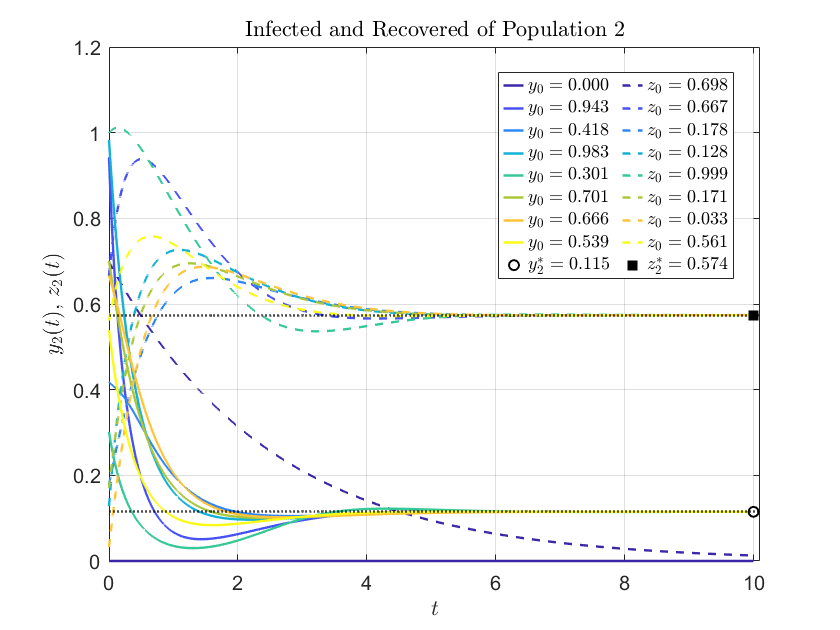}
	\caption{\label{fig:infected_rec_2} Time evolution of the infected and recovered individuals of population $2$ in the interaction network considered in Section \ref{sec:simulations} for $8$ different initial conditions. The black circle and square indicate the equilibrium values of the infected and recovered fractions, respectively. Solid and dashed lines represent infected (\(y_{2}\)) and recovered (\(z_{2}\)) fractions of individuals, respectively.}
\end{figure}
	
	\section{Conclusion}\label{sec:conclusion} 
We have studied the dynamic behavior of the network SIRS epidemic model on connected but otherwise general interaction graphs. We have proved that the dominant eigenvalue of a rescaled version of the interaction matrix, whose rows are normalized by the corresponding recovery rates, is the basic reproduction number $R_0$ and determines the qualitative behavior of the system. If $R_0\le1$, then the disease-free equilibrium is globally asymptotically stable; conversely, if $R_0 > 1$, then the disease-free equilibrium is unstable and an asymptotically stable endemic equilibrium emerges. We have also identified key monotonicity properties of the dependence of the endemic equilibrium on the model parameters  and obtained a provably convergent distributed iterative algorithm for its computation. 

Possible future research directions include proving almost global asymptotic stability of the endemic equilibrium and extending the analysis to non-connected networks. 
	\appendices

	\section{Proof of Lemma \ref{lemma:well_pos}}\label{app:well_pos}

	\begin{proof}
	Local existence and uniqueness of solutions follow directly from the Cauchy–Lipschitz theorem, since the vector field in \eqref{sirs_network_sistema} is locally Lipschitz. 
		From the fact that
		$\dot{x_i}(t)+\dot{y_i}(t)+\dot{z_i}(t)=0$ for all $t\ge0$ and  $i$ in $\mathcal V$, 
		we have that the set $\{x+y+z=\1\}$ is invariant. 
		Moreover, for every $i$ in $\V$, we have:
		\begin{itemize}
			\item if $x_i=0$, then $\dot x_i=\delta_i z_i \ge 0$;
			\item if $y_i=0$, then $\dot y_i = x_i \sum_j W_{ij}y_j \ge 0$;
			\item if $z_i=0$, then $\dot z_i = \gamma_i y_i \ge 0$. 
		\end{itemize}
		Therefore, no component $x_i,y_i,z_i$ can cross from a nonnegative value to a negative one. Hence, the nonnegative orthant $\{x\ge\0,y\ge\0,z\ge\0\}$ is positively invariant.
		
		Since the compact set $\X$ is the intersection: 
		$$\X=\{x+y+z=\1\}\cap\{x\ge\0,y\ge0,z\ge0\}\,,$$
		then it is positively invariant. 
%
%
%
Since $\mathcal X$ is compact, every solution starting in $\mathcal X$ remains in a compact subset of the domain of the vector field for as long as it exists. Hence, by the continuation theorem for ordinary differential equations, no finite-time blow-up can occur, and the unique local solution extends to all $t\ge0$.
	\end{proof}
	\section{Proof of Proposition \ref{theo:system_pos}}\label{app:lemma_positiveness}
    Proposition \ref{theo:system_pos}(i) follows from the observation that, if $y=\0$, then $\dot y=\0$. 
	The proof of Proposition \ref{theo:system_pos}(ii) is more involved and is based on the following technical result. 
	\begin{lemma}\label{lemma:positivness}
		Consider the network SIRS epidemic model \eqref{sirs_network_sistema} with irreducible interaction matrix $W$ in $\mathbb R_+^{\mathcal V\times \mathcal V}$ and positive recovery and loss-of-immunity rate vectors $\gamma$ and $\delta$ in $ \mathbb R_{++}^{\mathcal{V}}$. Consider an initial condition $(x(0),y(0),z(0))$ in $\X$ such that $y(0)\gneq\0$. Then, for every $i$ in $\V$ such that $y_i(0)=0$, there exists $p_i^*\ge1$ such that
		\begin{itemize}
			\item $y_i^{(p)}(0)=0$ for every $p\in\{1,\dots,p_i^*-1\}$;
			\item $y_i^{(p_i^*)}(0)>0$.
		\end{itemize}
	\end{lemma}
	
	\begin{proof}
		We first compute the $r$-th time derivative of $y_i$. Write $\dot{y}_i=f_ig_i-\gamma_i y_i$, where $f_i=\sum_j W_{ij}y_j$ and $g_i=1-y_i-z_i$. By Leibniz's rule,
		\[
		(f_ig_i)^{(r-1)}=\sum_{h=0}^{r-1} \binom{r-1}{h}f_i^{(h)}g_i^{(r-1-h)}\,,
		\]
		hence the $r$-th time derivative of $y_i$ is
		\be
		\label{formula_derivazione_k}
		y_i^{(r)} =
		\sum_{h=0}^{r-1} \tbinom{r-1}{h}
		\sum\nolimits_{k} W_{ik} y_k^{(h)} (1-y_i-z_i)^{(r-1-h)}
		- \gamma_i y_i^{(r-1)}\,.
		\ee
		
		For every pair of nodes $i,j$ in $\V$, let $\Omega_{ij}$ be the set of paths from $i$ to $j$. For $\omega\in\Omega_{ij}$, let $\alpha(\omega)$ be its length and
		\[
		m(\omega) := \bigl|\{ k \in \omega:k\ne j,\ z_k(0)=1 \}\bigr|\,.
		\]
		For $i$ in $\V$ such that $y_i(0)=0$, define
		\[
		\beta_i := \min_{j:\,y_j(0)>0}\ \min_{\omega\in\Omega_{ij}} \big(\alpha(\omega)+m(\omega)\big)\,.
		\]
		Since $W$ is irreducible, the associated graph is strongly connected, so $\Omega_{ij}\ne\emptyset$ for some $j$ with $y_j(0)>0$ (recall $y(0)\gneq \0$), hence $\beta_i<\infty$.
		
		Let $\varepsilon_i:=\chi_{\{z_i(0)=1\}}\in\{0,1\}$. We claim $\beta_i$ satisfies the recursion
		\be\label{eq:bellman}
		\beta_i = 1+\varepsilon_i+\min_{k:\,W_{ik}>0}\beta_k\,,
		\ee
		where $\beta_k:=0$ if $y_k(0)>0$. Indeed, take an optimal pair $(j,\omega)$ for $i$, and let $k$ be the node following $i$ on $\omega$ (so $W_{ik}>0$). Since $i\ne j$ (as $y_i(0)=0<y_j(0)$), $i$ contributes to $m(\omega)$ if and only if $\varepsilon_i=1$, and the sub-path $\omega'$ of $\omega$ from $k$ to $j$ satisfies $\alpha(\omega)+m(\omega)=1+\varepsilon_i+\alpha(\omega')+m(\omega')$; minimizing over $k$, $j$ and paths gives \eqref{eq:bellman}.
		
		If $y_i(0)=0$, then $g_i(0)=1-z_i(0)$, which is positive if $\varepsilon_i=0$ and zero if $\varepsilon_i=1$. If $\varepsilon_i=1$, then $y_i(0)=g_i(0)=0$ gives $\dot y_i(0)=f_i(0)g_i(0)-\gamma_i y_i(0)=0$, and since $\dot z_i=\gamma_i y_i-\delta_i z_i$, we get $\dot z_i(0)=-\delta_i$; therefore
		\[
		\dot g_i(0)=-\dot y_i(0)-\dot z_i(0)=\delta_i>0\,.
		\]
		In summary: $g_i^{(m)}(0)=0$ for $m<\varepsilon_i$, and $g_i^{(\varepsilon_i)}(0)>0$.
		
		We prove by induction on $r$ that, for every $i\in\V$ with $y_i(0)=0$ and $\beta_i\ge r$,
		\[
		y_i^{(r)}(0)=0 \ \text{ if } r<\beta_i\,,
		\qquad
		y_i^{(\beta_i)}(0)>0 \ \text{ if } r=\beta_i\,.
		\]
		
		If $r=0$, then $\beta_i\ge1>r$ and $y_i(0)=0$.
		
		Assume that the statement holds until $r-1$ and fix $i$ with $y_i(0)=0$ and $\beta_i\ge r\ge1$. Let $q_i:=\min_{k:W_{ik}>0}\beta_k$, so that $\beta_i=1+\varepsilon_i+q_i$ by \eqref{eq:bellman}. By \eqref{formula_derivazione_k}, $y_i^{(r)}(0)=T_1-T_2$, with
		\[
		T_1:=\sum_{h=0}^{r-1}\tbinom{r-1}{h}\Big(\sum_k W_{ik}y_k^{(h)}(0)\Big)g_i^{(r-1-h)}(0)\,,\]
		and
		\[
		T_2:=\gamma_i y_i^{(r-1)}(0)\,.
		\]
		Since $r-1<r\le\beta_i$, by the inductive hypothesis, $T_2=0$.
		
		Consider a term of $T_1$ indexed by $h\in\{0,\dots,r-1\}$, and set $m:=r-1-h$.
		\begin{itemize}
			\item If $h<q_i<\beta_i$: for every $k$ with $W_{ik}>0$ we have $\beta_k\ge q_i>h$, so $y_k^{(h)}(0)=0$;
			\item If $h\ge q_i$ and $r<\beta_i$ (i.e.\ $r\le q_i+\varepsilon_i$): then $m=r-1-h\le r-1-q_i\le \varepsilon_i-1<\varepsilon_i$, so $g_i^{(m)}(0)=0$.
		\end{itemize}
		Hence, if $r<\beta_i$, every term of $T_1$ vanishes, so $y_i^{(r)}(0)=T_1-T_2=0$.
		
		If $r=\beta_i=1+\varepsilon_i+q_i$, the terms of $T_1$ not already covered by the first bullet (i.e.\ with $h\ge q_i$) are $h=q_i$, and $h=q_i+1$ when $\varepsilon_i=1$:
		\begin{itemize}
			\item $h=q_i$: here $m=r-1-q_i=\varepsilon_i$, so $g_i^{(m)}(0)=g_i^{(\varepsilon_i)}(0)>0$. Moreover, since $q_i< r$, $y_k^{(q_i)}(0)=0$ whenever $\beta_k>q_i$ and $y_k^{(q_i)}(0)>0$ whenever $\beta_k=q_i$ (or $y_k(0)>0$ and $q_i=0$); since $K:=\{k: W_{ik}>0,\ \beta_k=q_i\}$ is nonempty by definition of $q_i$,
			\[
			\sum_k W_{ik}y_k^{(q_i)}(0)=\sum_{k\in K} W_{ik}y_k^{(q_i)}(0)>0\,,
			\]
			so this term of $T_1$ is strictly positive.
			\item $h=q_i+1$ (only if $\varepsilon_i=1$): here $m=r-1-h=0=\varepsilon_i-1$, so $g_i^{(0)}(0)=g_i(0)=0$; this term vanishes.
		\end{itemize}
		Thus $T_1>0$ and $T_2=0$, so $y_i^{(\beta_i)}(0)=T_1>0$, proving $H(r)$ also in this case.
		
		Taking $p_i^*=\beta_i$, we can conclude the proof. 
		\end{proof}

We can now prove Proposition \ref{theo:system_pos}(ii). 	
	Using \Cref{lemma:positivness} we obtain that for every $i$ in $ \mathcal{V}$, there exists $r_i$ in $\mathbb{N}$ such that
	\be
		y_i^{(r_i)}(0)>0, \quad \text{and} \quad y_i^{(\beta)}(0)\ge0 \ \text{for all } \beta<r_i\,.
	\ee
	This means that there exists a right neighborhood $(0,\epsilon_i)$ of the origin with $\epsilon_i>0$ such that $y_i(t)>0,\,\, \forall \, t\in (0,\epsilon_i)$. Let $\epsilon=\min_{i\in\mathcal{V}}\epsilon_i$, so that $y_i(t)>0$ for all $i$ in $\mathcal{V}$ and all $t\in(0,\epsilon)$. 
	Define
	\be\label{def_T} T=\sup \{t\ge 0: y(s)>0, \forall s\in(0,t)\}. \ee
	Assume by contradiction that $T<+\infty$. Hence $y(T)\ge\0$ with at least one zero component, by definition of $T$. Let $i^*$ in $\V$ be such that $y_{i^*}(T)=0.$ Since W is irreducible, there 
	exists $j$ in $\mathcal V$ such that $W_{i^*j}>0$. Moreover, by definition of $T$, $y_j(s)>0$ for every $s\in(0,T)$. Hence $(Wy(s))_{i^*}\ge W_{i^*j}y_j(s)$ for every $s\in(0,T)$.

	Moreover, it cannot be $x_{i^*}(s)=0$ for every $s\in(0,T)$. Indeed, we would have $0=\dot x_{i^*}(s) = \delta_{i^*}z_{i^*}(s)$, so that $z_{i^*}(s)=0$ and $y_{i^*}(s)=1$ for every $s\in(0,T)$. By continuity, it would be $y_{i^*}(T)=1$, which is a contradiction. Then, there must be $s^*\in(0,T)$ and a neighborhood $\mc U\subset (0,T)$ of $s^*$ such that $x_{i^*}(s)>0$ for every $s\in\mc U$. Then,
	\begin{align*}
			y_{i^*}(T)=&e^{-\gamma_{i^*}T} y_{i^*}(0)+\int_{0}^{T} e^{-\gamma_{i^*}(T-s)}x_{i^*}(s)(Wy(s))_{i^*}ds\\
				\ge& e^{-\gamma_{i^*}T} y_{i^*}(0)+\int_{0}^{T} e^{-\gamma_{i^*}(T-s)}x_{i^*}(s) W_{i^*j}y_j(s)ds\\
				\ge & e^{-\gamma_{i^*}T} y_{i^*}(0)+\int_{\mc U} e^{-\gamma_{i^*}(T-s)}x_{i^*}(s) W_{i^*j}y_j(s)ds\\
				>&0
	\end{align*}

	But this contradicts the hypothesis that $y_{i^*}(T)=0$. Hence, it must be $T=+\infty$, so that $y(t)>\0 ,\,\, \forall\, t>0.$

\section{Proof of Lemma \ref{equivalence_equilibrium_fixed}}\label{app:proof_equivalence}
\begin{proof}
	At equilibrium, the following  must be satisfied:
	\be	x^*=\1-y^*-z^*\qquad
		z^*=[\gamma][\delta]^{-1}y^*\label{equilibrium_eq_z}\,,\ee
		\be\left(I-[([\gamma]+[\delta])[\delta]^\inv y^*]\right)Wy^* =[\gamma] y^*\label{equilibrium_eq_y}\,.
	\ee
	From \eqref{equilibrium_eq_z}, it must be $(\1-y^*-\ga\de^\inv y^*,y^*,[\gamma][\delta]^\inv y^*)$ for some $y^*$ in $\mc Y_\alpha$.
Rewrite \eqref{equilibrium_eq_y} componentwise
	\be\label{equilibrium_eq_y_i}
		(Wy^*)_i/\gamma_i={y_i^*}/(1-(\gamma_i+\delta_i)y^*_i/{\delta_i})\,,
	\ee
	and apply $\Psi_i(\,\cdot,\de^\inv \ga)$ on both sides, obtaining
	\[
	y_i^*=\Psi_i(\ga^\inv Wy^*,\de^\inv \ga)=\Phi_i(y^*,\ga^\inv W,\de^\inv \gamma)\,.
	\]
	Hence, finding an equilibrium point for \eqref{sirs_network_sistema} is equivalent to finding a fixed point for $\Phi(\,\cdot\,,\ga^\inv W,\de^\inv \gamma)$. 
\end{proof}
	
	\section{Proof of \Cref{lem:H}}\label{app:lem_H}
(i) It follows directly from the fact that $\Psi_i(\0,\alpha)=0$, for every $i$ in $\mc V$.

(ii) Since, for every $i$ in $\mathcal V$, $\Psi_i(y,\alpha)<y_i$ for every $y_i>0$, then $\Psi(y,\alpha)< y$ for every $y>\mathbf0$. Then, by definition of $\Phi$,
			\[
			\Phi(y,M,\alpha)=\Psi(My,\alpha)<My
			\]
			for every $y>\mathbf 0$ and irreducible $M$.

(iii) Using the definition of $\Phi$,
			\[
			(\Phi(y,M,\alpha))_i=\dfrac{\sum_j M_{ij}y_j}{1+(1+\alpha_i)\sum_jM_{ij}y_j}
			\]
			we obtain
			\[
			\dfrac{\partial (\Phi(y,M,\alpha))_i}{\partial y_j}=\dfrac{M_{ij}}{(1+(1+\alpha_i)\sum_jM_{ij}y_j)^2}\ge0
			\]
			\[
			\dfrac{\partial (\Phi(y,M,\alpha))_i}{\partial M_{ij}}=\dfrac{y_j}{(1+(1+\alpha_i)\sum_jM_{ij}y_j)^2}\ge0
			\]
			\[
			\dfrac{\partial (\Phi(y,M,\alpha))_i}{\partial \alpha_i}=-\dfrac{(\sum_j M_{ij}y_j)^2}{(1+(1+\alpha_i)\sum_jM_{ij}y_j)^2}\le 0
			\]
			and this concludes the proof.\qed

	\section{Proof of \Cref{lemma:lyap_fun}}\label{app:lyap_fun}	Let $\bar w=\ga^\inv\bar v$. First, compute the time derivative 
		\be\label{lyap_deriv}
			\begin{split}
				\dot{V}=&\bar{w}^T\dot{y}\\
				=&\bar{w}^T(I-[y+z])Wy-\bar{w}^T\ga y\\
				=&\bar{w}^TWy-\bar{w}^T\ga y-\bar{w}^T[y+z]Wy\\
				=&\bar{v}^T\ga^\inv W y -\bar{v}^T\ga^\inv\ga y-\bar{w}^T[y+z]Wy\\
				=&(R_0-1)\bar{v}^Ty-\bar{w}^T[y+z]Wy \\ 
				=&(R_0-1)\bar{w}^T\ga y-\bar{w}^T[y+z]Wy\,.  \\
			\end{split}
		\ee
		If $R_0\le1$, then $\dot{V}\le-\bar{w}^T[y+z]Wy\le0$ and 
		$\dot{V}=0$ if and only if either $y=\0$ or $R_0=1$ and $\mathrm{supp}(Wy)\cap\mathrm{supp}(y+z)=\emptyset$. 
		In the latter case, for any $(y,z)$ belonging to an invariant subset of $\{\dot{V}=0\}$, consider the sets 
		$\mc J=\{j :y_j >0\}$, $\mc I=\{i:y_i=0\}$.
		Assume by contradiction that $\mc J\ne\emptyset$. Since $W$ is irreducible, there must exist $i$ in $\mc I$ and $j$ in $\mc J$ such that $W_{ij}>0$. Therefore, $i$ in $ \mathrm{supp}(Wy)$ which means that $i\notin \mathrm{supp}(y+z)$ and so $y_{i}=z_{i}=0$. We obtain
		$$\dot{y_{i}}=(1-y_{i}-z_{i})\sum\nolimits_k W_{ik}y_k -\gamma_{i} y_{i}\ge W_{ij}y_j>0\,.$$
		Therefore $\dot y_{i}>0$ at the considered state. By continuity, there exists $\varepsilon>0$ such that for $t\in(0,\varepsilon)$ we have $(Wy)_{i}(t)>0$ and $(y+z)_{i}(t)>0$. For such small $t$ the term $\bar w_{i}(y_{i}(t)+z_{i}(t))(Wy)_{i}(t)$ is strictly positive, hence
$$\dot V=-\sum\nolimits_i \bar w_i (y_i(t)+z_i(t))(Wy)_i(t) <0\,.$$
		This contradicts the assumption that the trajectory stays in the invariant subset of \(\{\dot{V}=0\}\). Thus no $y\ne\0$ can belong to an invariant subset of $\{\dot{V}=0\}$, and, since $\{y=\0\}$ is invariant by Proposition \ref{theo:system_pos}(i), we can conclude the proof. 
	\qed

	\section{Proof of \Cref{lemma:S_invertible}}\label{app:lemma_schur}
For every $k$ in $\mc V$, let $$H(\lambda)=S(\lambda) [y^*]\,,\qquad m_k=(Wy^*)_ky_k^*\,,$$ and $$g_k(\lambda)=\lambda y_k^*+\frac{\gamma_k m_k}{\delta_k+\lambda}\,.$$ 
		Then, $$H_{kk}(\lambda)=x_k^*W_{kk}y_k^*-\gamma_k y_k^* - m_k-g_k(\lambda)\,,\qquad\forall k\in\mc V\,,$$
		and  $$H_{kj}(\lambda)=x_k^*W_{kj}y_j^* \geq 0\,,\qquad \forall j\neq k\,.$$ 
Hence, the Gershgorin Theorem implies that
\be\label{Gershgorin}\sigma(H(\lambda))\subseteq		\bigcup_{k\in\mc V} B_{R_k}(H_{kk}(\lambda))\,,\ee
where
$$R_k=\sum_{j\neq k}H_{kj}(\lambda)=x_k^*\sum_{j\neq k} W_{kj}y_j^*\,,\qquad \forall k\,.$$ 
Now, observe that, if $\Re(\lambda)>-\eta$, then 
$$\ba{rcl}\Re(g_k(\lambda))
&=&\ds\Re(\lambda)y_k^*+\gamma_k m_k\Re\left(\frac1{\delta_k+\lambda}\right)\\[7pt]
&=&\ds\Re(\lambda)y_k^*+\frac{\gamma_k m_k(\delta_k+\Re(\lambda))}{(\delta_k+\Re(\lambda))^2+Im(\lambda)^2}\\[7pt]
			&>&\ds -\eta y_k^*+\frac{\gamma_k m_k (\delta_k-\eta)}{(\delta_k+\Re(\lambda))^2+Im(\lambda)^2}\\[7pt]
			&\ge &-m_k \,, 
		\ea$$
for every $k$ in $\mc V$, so that
\be\label{RHkk}\ba{rcl}\!\Re(H_{kk})\!+\!R_k
&\!\!\!=&\!\!\!\ds x_k^*\sum\nolimits_{j} W_{kj}y_j^*-\gamma_k y_k^*-m_{k}-\Re\left(g_k(\lambda)\right)\\[7pt]
&\!\!\!=&\!\!\!\ds-m_k-\Re\left(g_k(\lambda)\right)\\[7pt]
&\!\!\!<&\!\!\!0\,.\ea\ee
It then follows from \eqref{Gershgorin} and \eqref{RHkk} that 
$$\sigma(H(\lambda))\subseteq		\bigcup_{k\in\mc V} B_{R_k}(H_{kk}(\lambda))\subseteq\{\mu\in \mathbb{C}: \Re(\mu)<0\}\,.$$
%
%
%
%
		Hence, all eigenvalues of $H(\lambda)$ have negative real part, so that in particular $\operatorname{det}(H(\lambda))\neq 0 $. 
		Since $y^*>\0\,$, this implies that $ S(\lambda)=H(\lambda)[y^*]^{-1}$ is nonsingular. \qed

	\section{Region of attraction of the endemic equilibrium point with rank-one interaction matrix}\label{sec:rank-one}
For a rank-one interaction matrix, the Lyapunov argument  of \cite{Li2014} can be adapted to prove the following. 
\begin{proposition}\label{theo:stability_end_rank1}
Consider the network SIRS epidemic model \eqref{sirs_network_sistema} with rank-one interaction matrix $W=ab^T$, $a,b>\0$,  homogeneous recovery rate profile $\gamma=\bar\gamma \mathbf 1$, $\bar\gamma>0$, and loss-of-immunity rate profile $\delta>\0$. If $R_0=a^Tb/\bar\gamma>1$, then the endemic equilibrium $(x^*,y^*,z^*)$ is asymptotically stable and its region of attraction is $\mathcal{X}\setminus\{y=\mathbf 0\}$.
\end{proposition}
\begin{proof}
Let $h(y)=b^Ty$, $h^*= h(y^*)$. Define, for $y\gneq\mathbf0$:
	$$
		V_1=\frac{1}{2}\sum_i\frac{b_i}{x_i^*}(x_i-x_i^*)^2\,,\qquad	V_2=\frac{1}{2}\sum_i \frac{\delta_ib_i}{\bar \gamma x_i^*}(z_i-z_i^*)^2\,,$$
$$V_3=h(y)-h^*+h^*\ln\left({{h^*}/{h(y)}}\right)\,,\qquad 
				V=V_1+V_2+V_3\,.$$
$V$ is positive definite with respect to $(x^*,y^*,z^*)$. Moreover, 
$$\ba{rcl}\!\!\!\dot V_1\!\!\!
			&\!\!\!=\!\!\!&-\sum_i\frac{b_i}{x_i^*} a_i h (x_i-x_i^*)^2-(h-h^*)\sum_ib_i a_i(x_i-x_i^*)\\&&+\sum_i\frac{b_i}{x_i^*}\delta_i(x_i-x_i^*)(z_i-z_i^*)\,,\ea$$
$$\ba{rcl}\!\!\!\dot V_2\!\!\!
			&=&-\sum_i \delta_i\frac{b_i}{x_i^*}(z_i-z_i^*)(x_i-x_i^*)\\
			&&-\sum_i\frac{\delta_i}{\bar \gamma}\frac{b_i}{x_i^*}(\bar \gamma+ \delta_i)(z_i-z_i^*)^2\,,\ea$$
$$\ba{rcl}\!\!\!\dot V_3\!\!\!&=&(h-h^*)\dot h/{h}=(h-h^*)\sum_i b_ia_i(x_i-x_i^*)\,,\ea$$
$$
		\ba{rcl}\dot V&=&\dot V_1+\dot V_2+\dot V_3\\
		&=&\sum_i\left(-\frac{b_i a_i h}{x_i^*}(x_i-x_i^*)^2-\frac{\delta_i b_i (\bar \gamma+\delta_i)}{\bar \gamma x_i^*}(z_i-z_i^*)^2\right)\,.
	\ea$$
	Hence, $\dot V\le 0$ and $\dot V=0$ if and only if $(x,y,z)=(x^*,y^*,z^*)$. By LaSalle's invariance principle, $(x^*,y^*,z^*)$ is globally asymptotically stable on $\mathcal X\setminus\{y=\mathbf0\}$.
\end{proof}



		\bibliographystyle{ieeetr}
		\bibliography{bib}

\begin{thebibliography}{10}

\bibitem{Kermack.McKendrick:1927}
W.~O. Kermack and A.~G. McKendrick, ``A contribution to the mathematical theory
  of epidemics,'' {\em Proceedings of the {R}oyal {S}ociety of {L}ondon. Series
  A}, vol.~115, no.~772, pp.~700--721, 1927.

\bibitem{KermackMcKendrick1932}
W.~O. Kermack and A.~G. McKendrick, ``Contributions to the mathematical theory
  of epidemics. {II}. {T}he problem of endemicity,'' {\em Proceedings of the
  Royal Society A}, vol.~138, pp.~55--83, 1932.

\bibitem{anderson-may1991}
R.~M. Anderson and R.~M. May, {\em Infectious Diseases of Humans: Dynamics and
  Control}.
\newblock Oxford: Oxford University Press, 1991.

\bibitem{Diekmann2000MathematicalEO}
O.~Diekmann and J.~A.~P. Heesterbeek, {\em Mathematical Epidemiology of
  Infectious Diseases: Model Building, Analysis and Interpretation}.
\newblock Wiley, 2000.

\bibitem{cooke1973some}
K.~L. Cooke and J.~A. Yorke, ``Some equations modelling growth processes and
  gonorrhea epidemics,'' {\em Mathematical Biosciences}, vol.~16, no.~1-2,
  pp.~75--101, 1973.

\bibitem{hethcote1973asymptotic}
H.~W. Hethcote, ``Asymptotic behavior in a deterministic epidemic model,'' {\em
  Bulletin of Mathematical Biology}, vol.~35, pp.~607--614, 1973.

\bibitem{GarnettAnderson1996}
G.~P. Garnett and R.~M. Anderson, ``Sexually transmitted diseases and sexual
  behavior: Insights from mathematical models,'' {\em The Journal of Infectious
  Diseases}, vol.~174, pp.~S150--S161, October 1996.

\bibitem{liu1986influence}
W.-m. Liu, S.~A. Levin, and Y.~Iwasa, ``Influence of nonlinear incidence rates
  upon the behavior of {SIRS} epidemiological models,'' {\em Journal of
  Mathematical Biology}, vol.~23, no.~2, pp.~187--204, 1986.

\bibitem{Jin2007}
Y.~Jin, W.~Wang, and S.~Xiao, ``An {SIRS} model with a nonlinear incidence
  rate,'' {\em Chaos, Solitons \& Fractals}, vol.~34, pp.~1482--1497, December
  2007.

\bibitem{KOROBEINIKOV2002955}
A.~Korobeinikov and G.~Wake, ``Lyapunov functions and global stability for
  {SIR}, {SIRS}, and {SIS} epidemiological models,'' {\em Applied Mathematics
  Letters}, vol.~15, no.~8, pp.~955--960, 2002.

\bibitem{lajmanovich1976deterministic}
A.~Lajmanovich and J.~A. Yorke, ``A deterministic model for gonorrhea in a
  nonhomogeneous population,'' {\em Mathematical Biosciences}, vol.~28,
  no.~3-4, pp.~221--236, 1976.

\bibitem{HirschMonotoneSystems06}
M.~Hirsch and H.~Smith, ``Monotone dynamical systems,'' vol.~2 of {\em Handbook
  of Differential Equations: Ordinary Differential Equations, Chapter 4},
  pp.~239 -- 357, North-Holland, 2006.

\bibitem{pastor2015epidemic}
R.~Pastor-Satorras, C.~Castellano, P.~Van~Mieghem, and A.~Vespignani,
  ``Epidemic processes in complex networks,'' {\em Reviews of Modern Physics},
  vol.~87, pp.~925--979, 2015.

\bibitem{pare2020modeling}
P.~E. Par{\'e}, C.~L. Beck, and T.~Ba{\c s}ar, ``Modeling, estimation, and
  analysis of epidemics over networks: An overview,'' {\em Annual Reviews in
  Control}, vol.~50, pp.~345--360, 2020.

\bibitem{Alutto.ea:2025}
M.~Alutto, L.~Cianfanelli, G.~Como, and F.~Fagnani, ``On the dynamic behavior
  of the network {SIR} epidemic model,'' {\em IEEE Transactions on Control of
  Network Systems}, vol.~12, no.~1, pp.~177--189, 2025.

\bibitem{Nowzari.ea:2016}
C.~Nowzari, V.~M. Preciado, and G.~J. Pappas, ``Analysis and control of
  epidemics: A survey of spreading processes on complex networks,'' {\em IEEE
  Control Systems Magazine}, vol.~36, no.~1, pp.~26--46, 2016.

\bibitem{Mei.ea:2017}
W.~Mei, S.~Mohagheghi, S.~Zampieri, and F.~Bullo, ``On the dynamics of
  deterministic epidemic propagation over networks,'' {\em Annual Reviews in
  Control}, vol.~44, pp.~116--128, 2017.

\bibitem{Brauer2017}
F.~Brauer, ``Mathematical epidemiology: Past, present, and future,'' {\em
  Infectious Disease Modelling}, vol.~2, no.~2, pp.~113--127, 2017.

\bibitem{ZinoCao2021}
L.~Zino and M.~Cao, ``Analysis, prediction, and control of epidemics: A survey
  from scalar to dynamic network models,'' {\em IEEE Circuits and Systems
  Magazine}, vol.~21, no.~2, pp.~4--24, 2021.

\bibitem{Li2014}
C.~Li, C.~Tsai, and S.~Yang, ``Analysis of epidemic spreading of an {SIRS}
  model in complex heterogeneous networks,'' {\em Communications in Nonlinear
  Science and Numerical Simulations}, vol.~19, pp.~1042--1054, 2014.

\bibitem{chen2014global}
L.~Chen and J.~Sun, ``Global stability and optimal control of an {SIRS}
  epidemic model on heterogeneous networks,'' {\em Physica A: Statistical
  Mechanics and its Applications}, vol.~410, pp.~196--204, 2014.

\bibitem{zhang2021layered}
Y.~Zhang and D.~Pan, ``Layered {SIRS} model of information spread in complex
  networks,'' {\em Applied Mathematics and Computation}, vol.~411, p.~126524,
  2021.

\bibitem{saif2019epidemic}
M.~A. Saif, ``Epidemic threshold for the {SIRS} model on the networks,'' {\em
  Physica A: Statistical Mechanics and its Applications}, vol.~535, p.~122251,
  2019.

\bibitem{lijun-nonmonotone-2018}
L.~Liu, X.~Wei, and N.~Zhang, ``Global stability of a network-based {SIRS}
  epidemic model with nonmonotone incidence rate,'' {\em Physica A}, vol.~515,
  pp.~587--599, 2019.

\bibitem{yang2017heterogeneous}
L.~Yang, M.~Draief, and X.~Yang, ``Heterogeneous virus propagation in networks:
  A theoretical study,'' {\em Mathematical Methods in the Applied Sciences},
  vol.~40, no.~5, pp.~1396--1413, 2017.

\bibitem{MUROYA20131693}
Y.~Muroya, Y.~Enatsu, and T.~Kuniya, ``Global stability for a multi-group
  {SIRS} epidemic model with varying population sizes,'' {\em Nonlinear
  Analysis: Real World Applications}, vol.~14, no.~3, pp.~1693--1704, 2013.

\bibitem{berman1994nonnegative}
A.~Berman and R.~J. Plemmons, {\em Nonnegative Matrices in the Mathematical
  Sciences}.
\newblock Classics in Applied Mathematics, Philadelphia: SIAM, 1994.

\end{thebibliography}
		
	\end{document}